\documentclass[a4paper]{article}
\usepackage[a4paper,includeheadfoot,margin=2.54cm]{geometry}

\usepackage[T1]{fontenc}
\usepackage{amsthm}
\usepackage{amsmath}
\usepackage{amssymb}
\usepackage[affil-it]{authblk}
\usepackage{microtype}
\usepackage{cite}
\usepackage{hyperref}
\hypersetup{%
   breaklinks,%
   colorlinks=true,%
   linkcolor=[rgb]{0.0,0.0,0.0},%
   citecolor=[rgb]{0,0,0.45}
}
\usepackage[noabbrev,capitalise]{cleveref}
\usepackage{soul}
\usepackage{xspace}
\usepackage{caption} 
\usepackage{color}
\usepackage{mleftright}%
\usepackage{xspace}%
\usepackage{enumerate}
\usepackage{etoolbox}
\usepackage{lmodern}
\usepackage{algorithm}
\usepackage[noend]{algpseudocode}
\makeatletter
\patchcmd{\@maketitle}{\LARGE \@title}{\fontsize{16}{19.2}\selectfont\@title}{}{}\makeatother

\usepackage{graphicx}
\graphicspath{{figs/}}

\theoremstyle{plain}
\newtheorem{theorem}{Theorem}[section]
\newtheorem{lemma}[theorem]{Lemma}
\newtheorem{corollary}[theorem]{Corollary}

\DeclareMathOperator{\dist}{dist} 
\DeclareMathOperator{\rad}{rad} \DeclareMathOperator{\vio}{vio}
\DeclareMathOperator{\argmin}{argmin}
 \def\D{{\cal D}} 
\def\C{{\cal C}} \def\eps{{\varepsilon}} \def\Rr{{\cal R}} \def\A{{\cal
      A}} \def\B{{\cal B}}  \def\H{{\cal H}} \def\Y{{\cal Y}}
\def\W{{\cal W}} \def\R{{\mathbb R}} 
\def\eps{{\varepsilon}}

\newcommand{\Stacho}{Stach\'o\xspace}%
\newcommand{\Grunbaum}{Gr\"unbaum\xspace}
\providecommand{\Matousek}{Matou{\v s}ek\xspace}

\newlength{\savedparindent}

\title{Stabbing Pairwise Intersecting Disks by Five Points\footnote{
A preliminary version appeared as
S.~Har-Peled, H.~Kaplan, W.~Mulzer, L.~Roditty, P.~Seiferth,
M.~Sharir, and M.~Willert.
\emph{Stabbing Pairwise Intersecting Disks by Five Points.}
Proc.~29th ISAAC, pp.~50:1--50:12.
SHP was supported by a NSF AF awards CCF-1421231, and CCF-1217462.
WM was supported by DFG grant MU/3501/1 and ERC STG 757609.
PS was supported by DFG grant MU/3501/1.
MS was supported by ISF grant 892/13 and 260/18, by the Israeli Centers
of Research Excellence (I-CORE) program (Center No.~4/11), and by the
Blavatnik Research Fund in Computer Science at Tel Aviv University.
HK was supported by ISF grant 1595-19 and the Blavatnik Family Foundation.
Work on this paper was supported in part by grant 1367/2016 and 1161/2011 
from the German-Israeli Science Foundation (GIF).
}}

\author[1]{Sariel Har-Peled}
\author[2]{Haim Kaplan}
\author[3]{Wolfgang Mulzer} 
\author[4]{Liam Roditty}
\author[3]{\\Paul Seiferth} 
\author[2]{Micha Sharir} 
\author[3]{Max Willert}

\affil[1]{Department of Computer Science, University of Illinois, 
Urbana, IL 61801, USA\\
\texttt{sariel@illinois.edu}}

\affil[2]{School of Computer Science, Tel Aviv University, 
Tel~Aviv 69978, Israel\\
\texttt{\{haimk,michas\}@tau.ac.il}}

\affil[3]{Institut f\"ur Informatik, Freie Universit\"at Berlin, 
14195 Berlin, Germany\\
\texttt{\{mulzer,pseiferth,willerma\}@inf.fu-berlin.de}}

\affil[4]{Department of Computer Science, Bar Ilan University, 
Ramat Gan 5290002, Israel\\
\texttt{liamr@macs.biu.ac.il}}

\date{}
\begin{document}
\maketitle

\begin{abstract}
    Suppose we are given a set $\mathcal{D}$ of $n$ pairwise
    intersecting disks in the plane. A planar point set $P$
    \emph{stabs} $\mathcal{D}$ if and only if each disk in
    $\mathcal{D}$ contains at least one point from $P$. We present a
    deterministic algorithm that takes $O(n)$ time to find five points
    that stab $\mathcal{D}$.  Furthermore, we give a simple example of
    13 pairwise intersecting disks that cannot be stabbed by three
    points. Moreover, we present a simple argument showing that
    eight disks can be stabbed by at most three points.

    This provides a simple---albeit slightly weaker---algorithmic
    version of a classical result by Danzer that such a set
    $\mathcal{D}$ can always be stabbed by four points.
\end{abstract}

\section{Introduction}

The \emph{maximum clique problem} is a classic problem in combinatorial
optimization~\cite{Karp72}: given a simple graph $G = (V, E)$, 
find a maximum-cardinality 
set $C \subseteq V$ of vertices such that
any two distinct vertices in $C$ are adjacent.
In 1972, Karp proved that the maximum 
clique problem is NP-hard~\cite{Karp72}.
Even worse, a subsequent line of research
showed that the maximum clique problem is hard to
approximate. In particular,
we now know that for any fixed $\eps > 0$, if there is a polynomial-time
algorithm that approximates maximum clique in an $n$-vertex graph 
up to a factor of $n^{1-\eps}$, then 
$\text{P} = \text{NP}$~\cite{zuckerman2006linear}

However, if the
input graph has additional structure, the problem
can become easier. For example, if the input is the intersection graph of 
a set of disks in the plane, the maximum clique problem admits 
efficient (approximation)
algorithms: for unit disk graphs, 
it can be solved in polynomial time~\cite{clark1990unit}, 
while for general disk intersection graphs, there is a 
randomized EPTAS~\cite{bonamy2018eptas}. Earlier, Amb\"uhl 
and Wagner~\cite{ambuhl2005clique} presented a polynomial-time 
algorithm that computes
a $\tau/2$-approximation for the maximum clique in a general disk 
intersection graph, where $\tau$ is the minimum \emph{stabbing number} of
any arrangement of  pairwise intersecting disks in the plane, i.e., 
the minimum number of points that are needed to stab every disk in
such an arrangement.
Motivated by this
application,
our goal here is to understand this stabbing number better. 

Let $\D$ be a set of $n$ disks in the plane. If every \emph{three}
disks in $\D$ intersect, then Helly's theorem shows that the whole
intersection $\bigcap\D$ of $\D$ is
nonempty~\cite{Helly23,Helly30,Radon21}.  In other words, there is a
single point $p$ that lies in all disks of $\D$, that is, $p$
\emph{stabs} $\D$. More generally, when we know only that every
\emph{pair} of disks in $\D$ intersect, there must be a point set $P$
of constant size such that each disk in $\D$ contains at least one
point in $P$ -- the minimum cardinality of $P$ is the \emph{stabbing number}
of $\D$. It is indeed not surprising that $\D$ can be stabbed by a 
constant number of points,
but for some time, the exact bound remained elusive. Eventually,
in July 1956 at an Oberwolfach seminar, Danzer presented the answer:
four points are always sufficient and sometimes necessary to stab any
finite set of pairwise intersecting disks in the plane.
Danzer was not satisfied with his original
argument, so he never formally published it. In 1986, he presented a
new proof~\cite{Danzer86}. Previously, in 1981, \Stacho had already
given an alternative proof~\cite{Stacho81}, building on a previous
construction of five stabbing points~\cite{Stacho65}. This line of
work was motivated by a result of Hadwiger and Debrunner, who showed
that three points suffice to stab any finite set of pairwise
intersecting \emph{unit} disks~\cite{HadwigerDe55}.  In later work,
these results were significantly generalized and extended, culminating
in the celebrated $(p, q)$-theorem that was proven by Alon and
Kleitman in 1992~\cite{AlonKl92}.
See also a recent paper by
Dumitrescu and Jiang that studies generalizations of the stabbing
problem for translates and homothets of a convex
body~\cite{DumitrescuJi11}.

Danzer's published proof~\cite{Danzer86} is fairly involved.
It uses a compactness argument that does not seem to be 
constructive, and one part of the argument relies on
an underspecified 
verification by computer. Therefore, it is
quite challenging to check the correctness
of the argument, let alone to derive any intuition
from it.
There seems to be no obvious way to turn it into
an efficient algorithm for finding a stabbing set of size four. The
proof of \Stacho~\cite{Stacho81} is simpler, but
it is obtained through a lengthy case analysis that requires
a very disciplined and focused reader. Here, we present a new
argument that yields five stabbing points. Our proof is constructive,
and it lets us find the stabbing set in deterministic linear time.
Following the conference version of this paper,
Carmi, Katz, and Morin published a manuscript
in which they present an algorithm that can find four 
stabbing points
in linear time~\cite{CarmiKaMo18}.

As for lower bounds, \Grunbaum gave an example of 21 pairwise
intersecting disks that cannot be stabbed by three
points~\cite{Grunbaum59}. Later, Danzer reduced the number of disks to
ten~\cite{Danzer86}. This example is close to optimal, because every
set of eight disks can be stabbed by three points, as mentioned 
by \Stacho~\cite{Stacho65} and
formally proved in Section~\ref{sec:simple_bounds} below. However, it
is hard to verify Danzer's lower bound example---even with dynamic
geometry software, the positions of the disks cannot be visualized
easily.

We present a new and simple proof that shows 
that the stabbing number of $\D$ is 
upper bounded by $5$. Moreover, we obtain a linear time algorithm 
that can find these $5$ stabbing points. Finally, we present a simple 
construction of $13$ pairwise intersecting disks that cannot be stabbed
by $3$ points, and work out a proof of \Stacho's eight-disk claim.
 
\section{The Geometry of Pairwise Intersecting Disks}
\label{sec:geometry}

Let $\D$ be a set of $n$ pairwise intersecting disks in the plane.  A
disk $D_i \in \D$ is given by its center $c_i$ and its radius
$r_i$.
To simplify the analysis, we make the following assumptions:
(i) the radii of the disks are pairwise distinct; (ii) the
intersection of any two disks has a nonempty interior; and (iii) the
intersection of any three disks is either empty or has a nonempty
interior.  A simple perturbation argument can then handle the
degenerate cases.

\begin{figure}
\begin{center}
\includegraphics[page=1]{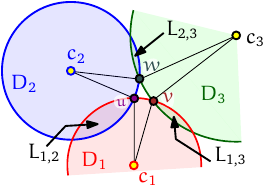}%
\hspace*{1cm}%
\includegraphics{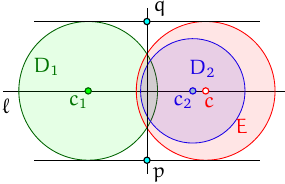}
\end{center}
\caption[Illustration \cref{lem:triple} and \cref{lem:subset}]{Left: At 
least one lens angle is large.
Right: $D_1$ and $E$ have the same radii and lens
angle $2\pi/3$. By \cref{lem:subset}, $D_2$ is a subset of
 $E$. $\{c_1,c,p,q\}$ is the set $P$ from \cref{lem:C:D:diff:size}.}
\label{fig:3_lens}%
\end{figure}

The \emph{lens} of two disks $D_i, D_j \in \D$ is the set
$L_{i, j} = D_i \cap D_j$.  Let $u$ be any of the two intersection
points of the boundary of $D_i$ and the boundary of $D_j$.  
The angle $\angle c_iuc_j$ is
called the \emph{lens angle} of $D_i$ and $D_j$. It is at most
$\pi$. A finite set $\mathcal{C}$ of disks is \emph{Helly} if their
common intersection $\bigcap \mathcal{C}$ is nonempty. Otherwise,
$\mathcal{C}$ is \emph{non-Helly}.  We present some useful geometric
lemmas.

\begin{lemma}%
    \label{lem:triple}%
    Let $\{D_1, D_2, D_3\}$ be a set of three pairwise intersecting
    disks that is non-Helly.  Then, the set contains two disks with
    lens angle larger than $2 \pi/3$.
\end{lemma}
\begin{proof}    
    Since $\{D_1, D_2, D_3\}$ is non-Helly, the lenses $L_{1, 2}$,
    $L_{1, 3}$ and $L_{2, 3}$ are pairwise disjoint. Let $u$ be the
    vertex of $L_{1, 2}$ nearer to $D_{3}$, and let $v$, $w$ be the
    analogous vertices of $L_{1,3}$ and $L_{2,3}$ (see
    \cref{fig:3_lens}, left).  Consider the simple hexagon
    $c_1 u c_2 w c_3v$, and write $\angle u$, $\angle v$, and
    $\angle w$ for its interior angles at $u$, $v$, and $w$. The sum
    of all interior angles is $4\pi$. Thus,
    $\angle u + \angle v + \angle w < 4\pi$, so at least one angle is
    less than $4\pi/3$.  It follows that the corresponding lens angle,
    which is the exterior angle at $u$, $v$, or $w$ must be larger 
    than $2\pi/3$.
\end{proof}

\begin{lemma}%
    \label{lem:subset}%
    Let $D_1$ and $D_2$ be two intersecting disks with $r_1 \geq r_2$
    and lens angle at least $2 \pi / 3$.  Let $E$ be the unique disk
    with radius $r_1$ and center $c$, such that
    \begin{enumerate}
    \item[(i)] the centers $c_1$,
    $c_2$, and $c$ are collinear and $c$ lies on the same side of
    $c_1$ as $c_2$; and 
    \item[(ii)] the lens angle of $D_1$ and $E$ is
    exactly $2\pi/3$ (see \cref{fig:3_lens}, right).
    \end{enumerate}Then, if $c_2$
    lies between $c_1$ and $c$, we have $D_2 \subseteq E$.
\end{lemma}%

\begin{proof}
    Let $x \in D_2$. Since $c_2$ lies between $c_1$ and $c$, the
    triangle inequality gives
    \begin{equation}%
        \label{eq:x:c}%
        |xc| \leq |xc_2| + |c_2c| = |xc_2| + |c_1c| - |c_1c_2|.
    \end{equation}
    Since $x \in D_2$, we get $|xc_2| \leq r_2$.  Also, since $D_1$
    and $E$ have radius $r_1$ each and lens angle $2\pi/3$, it follows
    that $|c_1c| = \sqrt{3} \, r_1$.  Finally,
    $|c_1c_2| = \sqrt{r_1^2+r_2^2-2r_1r_2\cos\alpha}$, by the law of
    cosines, where $\alpha$ is the lens angle of $D_1$ and $D_2$. As
    $\alpha \geq 2\pi/3$ and $r_1 \geq r_2$, we get
    \begin{math}
        \cos\alpha%
        \leq%
        -1/2 %
        = (\sqrt{3}-3/2)-\sqrt{3}+1%
        \leq%
        (\sqrt{3}-3/2){r_1}/{r_2}-\sqrt{3}+1,
    \end{math}
    As such, we have
    \begin{align*}
      |c_1 c_2|^2%
      &=%
        r_1^2 + r_2^2 - 2r_1r_2\cos\alpha
        \geq%
        r_1^2 + r_2^2 - 2r_1r_2 \Bigl( \bigl(\sqrt{3} - 3/2 \bigr)
        \frac{r_1}{r_2} - \sqrt{3} + 1\Bigr)%
      \\&%
      =%
      r_1^2 -  2\bigl(\sqrt{3} - 3/2 \bigr)
      r_1^2 +2 ( - \sqrt{3} + 1)r_1 r_2 + r_2^2 %
      \\&%
      =%
      (1 -  2\sqrt{3} + 3  )r_1^2 +2( - \sqrt{3} + 1)r_1 r_2  + r_2^2 %
      =%
      \bigl(r_1(\sqrt{3}-1)+r_2 \bigr)^2.
    \end{align*}
    Plugging this into Equation~\ref{eq:x:c} gives
    $|xc| \leq r_2 +
    \sqrt{3}r_1-(r_1\left(\sqrt{3}-1)+r_2\right)=r_1$, i.e.,
    $x \in E$.
\end{proof}

\begin{lemma}%
    \label{lem:D:E}%
    Let $D_1$ and $D_2$ be two intersecting disks with equal radius
    $r$ and lens angle $2\pi/3$. There is a set $P$ of four points so
    that any disk $F$ of radius at least $r$ that intersects both
    $D_1$ and $D_2$ contains a point of $P$.
\end{lemma}

\begin{proof}   
    Consider the two tangent lines of $D_1$ and $D_2$, and let $p$ and
    $q$ be the midpoints on these lines between the respective two
    tangency points. We set $P = \{ c_1, c_2, p, q\}$; see
    \cref{fig:P:covers}.
    
\begin{figure}
    \centerline{
       \includegraphics[page=2]{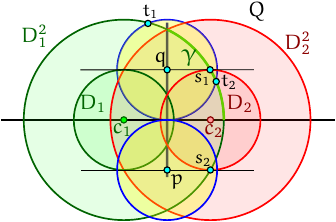}%
    }
    \caption[Illustration\cref{lem:D:E} and 
    \cref{lem:C:D:diff:size}]{Left: $P=\{c_1,c_2,p,q\}$
    is the stabbing set.
       The green arc $\gamma=\partial D_1^2 \cap Q$ is covered
       by $D_2 \cup D_q$.}
    \label{fig:P:covers}%
\end{figure}
 
    Given the disk $F$ that intersects both $D_1$ and $D_2$, we shrink
    its radius, keeping its center fixed, until either the radius
    becomes $r$ or until $F$ is tangent to $D_1$ or $D_2$. Suppose the
    latter case holds and $F$ is tangent to $D_1$. We move the center
    of $F$ continuously along the line spanned by the center of $F$
    and $c_1$ towards $c_1$, decreasing the radius of $F$ to maintain
    the tangency. We stop when either the radius of $F$ reaches $r$ or
    $F$ becomes tangent to $D_2$.  We obtain a disk $G \subseteq F$
    with center $c = (c_x, c_y)$ so that either: (i)
    $\text{radius}(G) = r$ and $G$ intersects both $D_1$ and $D_2$; or
    (ii) $\text{radius}(G) \geq r$ and $G$ is tangent to both $D_1$
    and $D_2$.  Since $G \subseteq F$, it suffices to show that
    $G \cap P \neq \emptyset$.
    
    We introduce a coordinate system, setting the origin $o$ midway
    between $c_1$ and $c_2$, so that the $y$-axis passes through $p$
    and $q$.  Then, as in \cref{fig:P:covers}, we have
    $c_1 = (-\sqrt{3}\,r / 2, 0)$, $c_2 = (\sqrt{3}\,r / 2, 0)$,
    $q = (0, r)$, and $p = (0, -r)$.

    For case (i), let $D_1^2$ be the disk of radius $2r$ centered at
    $c_1$, and $D_2^2$ the disk of radius $2r$ centered at
    $c_2$. Since $G$ has radius $r$ and intersects both $D_1$ and
    $D_2$, its center $c$ has distance at most $2r$ from both $c_1$
    and $c_2$, i.e., $c \in D_1^2 \cap D_2^2$.  Let $D_p$ and $D_q$ be
    the two disks of radius $r$ centered at $p$ and $q$. We will show
    that $D_1^2 \cap D_2^2 \subseteq D_1 \cup D_2 \cup D_p \cup D_q$.
    Then it is immediate that $G \cap P \neq \emptyset$.  By symmetry,
    it is enough to focus on the upper-right quadrant
    $Q = \{(x,y) \mid x \geq 0, y \geq 0\}$.  We show that all points
    in $D_1^2 \cap Q$ are covered by $D_2 \cup D_q$.  Without loss of
    generality, we assume that $r = 1$.  Then, the two intersection
    points of $D_1^2$ and $D_q$ are
    $t_1 = (\frac{5\sqrt{3} - 2\sqrt{87}}{28}, \frac{38 +
       3\sqrt{29}}{28}) \approx (-0.36, 1.93)$ and
    $t_2 = (\frac{5\sqrt{3} + 2\sqrt{87}}{28}, \frac{38 -
       3\sqrt{29}}{28}) \approx (0.98, 0.78)$, and the two
    intersection points of $D_1^2$ and $D_2$ are
    $s_1 = (\frac{\sqrt{3}}{2}, 1) \approx (0.87, 1)$ and
    $s_2 = (\frac{\sqrt{3}}{2}, -1) \approx (0.87, -1)$.  Let $\gamma$
    be the boundary curve of $D_1^2$ in $Q$.  Since
    $t_1, s_2 \not\in Q$ and since $t_2 \in D_2$ and $s_1 \in D_q$, it
    follows that $\gamma$ does not intersect the boundary of
    $D_2 \cup D_q$ and hence $\gamma \subset D_2 \cup D_q$.
    Furthermore, the subsegment of the $y$-axis from $o$ to the start
    point of $\gamma$ is contained in $D_q$, and the subsegment of the
    $x$-axis from $o$ to the endpoint of $\gamma$ is contained in
    $D_2$. Hence, the boundary of $D_1^2 \cap Q$ lies completely in
    $D_2 \cup D_q$, and since $D_2 \cup D_q$ is simply connected, it
    follows that $D_1^2 \cap Q \subseteq D_2 \cup D_q$, as desired.

    For case (ii), since $G$ is tangent to $D_1$ and $D_2$, the center
    $c$ of $G$ is on the perpendicular bisector of $c_1$ and $c_2$, so
    the points $p$, $o$, $q$ and $c$ are collinear.  Suppose without
    loss of generality that $c_y\geq 0$. Then, it is easily checked
    that $c$ lies above $q$, and
    $\text{radius}(G) + r= |c_1c| \geq |oc|=r + |qc|$, so $q \in G$.
\end{proof}
\begin{figure}\begin{center}

       \includegraphics[page=2]{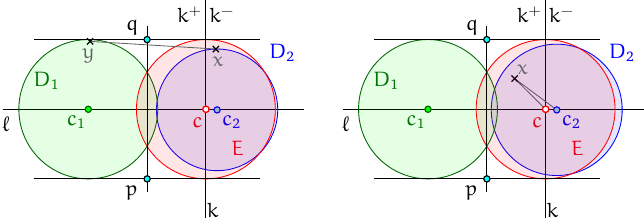}
\end{center}
    \caption[Illustration \cref{lem:C:D:diff:size}]{Proof of 
    \cref{lem:C:D:diff:size}. Left (Case (i)): $x$ is
       an arbitrary point in $D_2\cap F\setminus k^+$ and $y$ is an 
       arbitrary point in $D_1\cap F$. Right (Case (ii)): $x$ is
       an arbitrary point in $D_2\cap F\cap k^+$.  The angle at $c$ in
       the triangle $\Delta xcc_2$ is $\geq\pi/2$.}
    \label{fig:C:D:diff:size}%
\end{figure}
\begin{lemma}
    \label{lem:C:D:diff:size}%
    Consider two intersecting disks $D_1$ and $D_2$ with $r_1\geq r_2$
    and lens angle at least $2\pi/3$.  Then, there is a set $P$ of
    four points such that any disk $F$ of radius at least $r_1$ that
    intersects both $D_1$ and $D_2$ contains a point of $P$.
\end{lemma}

\begin{proof}
    Let $\ell$ be the line through $c_1$ and $c_2$.  Let $E$ be the
    disk of radius $r_1$ and center $c \in \ell$ that satisfies the
    conditions (i) and (ii) of \cref{lem:subset}. Let
    $P = \{c_1, c, p, q\}$ as in the proof of \cref{lem:D:E}, with
    respect to $D_1$ and $E$ (see \cref{fig:3_lens}, right).  We claim
    that
    \[
        D_1\cap F\neq\emptyset\ \wedge\ D_2\cap F\neq\emptyset\
        \Rightarrow\ E\cap F\neq\emptyset.  \tag{*}
    \]
    Once (*) is established, we are done by \cref{lem:D:E}.  If
    $D_2\subseteq E$, then (*) is immediate, so assume that
    $D_2 \not\subseteq E$.  By \cref{lem:subset}, $c$ lies between $c_1$
    and $c_2$.  Let $k$ be the line through $c$ perpendicular to
    $\ell$, and let $k^+$ be the open halfplane bounded by $k$ with
    $c_1 \in k^+$ and $k^-$ the open halfplane bounded by $k$ with
    $c_1 \not\in k^-$.  Since $|c_1c| = \sqrt{3}\,r_1 > r_1$, we have
    $D_1 \subset k^+$; see \cref{fig:C:D:diff:size}.  Recall that $F$
    has radius at least $r_1$ and intersects $D_1$ and $D_2$. We
    distinguish two cases: (i) there is no intersection of $F$ and
    $D_2$ in $k^+$, and (ii) there is an intersection of $F$ and $D_2$
    in $k^+$; see \cref{fig:C:D:diff:size} for the two cases.
    
    For case (i), let $x$ be any point in $D_1 \cap F$.  Since we know
    that $D_1 \subset k^+$, we have $x\in k^+$.  Moreover, let $y$ be
    any point in $D_2 \cap F$. By assumption, $y$ is not in $k^+$,
    but it must be in the infinite strip defined by the two tangents
    of $D_1$ and $E$. Thus, the line segment $\overline{xy}$
    intersects the diameter segment $k\cap E$.  Since $F$ is convex,
    the intersection of $\overline{xy}$ and $k\cap E$ is in $F$, so
    $E \cap F \neq \emptyset$.
    
    For case (ii), fix $x \in D_2\cap F \cap k^+$ arbitrarily.
    Consider the triangle $\Delta xcc_2$. Since $x \in k^+$, the angle
    at $c$ is at least $\pi/2$. Thus,
    $|xc| \leq |xc_2|$. Also, since $x \in D_2$, we know that
    $|xc_2| \leq r_2 \leq r_1$. Hence, $|xc| \leq r_1$, so $x \in E$
    and (*) follows, as $x \in E \cap F$.
\end{proof}

\section{Existence of Five Stabbing Points}

With these tools we can now show that
there is a stabbing set with five points.

\begin{theorem}%
    \label{thm:existence}%
    Let $\D$ be a set of $n$ pairwise intersecting disks in the plane.
    There is a set $P$ of five points such that each disk in $\D$
    contains at least one point from $P$.
\end{theorem}

\begin{proof}
    If $\D$ is Helly, there is a single point that lies in all disks
    of $\D$. Thus, assume that $\D$ is non-Helly, and let
    $D_1, D_2, \dots, D_n$ be the disks in $\D$ ordered by increasing
    radius. Let $i^*$ be the smallest index with
    $\bigcap_{i \leq i^*} D_i = \emptyset$. By Helly's
    theorem~\cite{Helly23,Helly30,Radon21}, there are indices
    $j, k < i^*$ such that $\{D_{i^*}, D_j, D_k\}$ is non-Helly. By
    \cref{lem:triple}, two disks in $\{D_{i^*}, D_j, D_k\}$ have
    lens angle at least $2\pi/3$.  Applying
    \cref{lem:C:D:diff:size} to these two disks, we obtain a set
    $P'$ of four points so that every disk $D_i$ with $i \geq i^*$
    contains at least one point from $P'$.  Furthermore, by definition
    of $i^*$, we have $\bigcap_{i < i^*} D_i \neq \emptyset$, so there
    is a point $q$ that stabs every disk $D_i$ with $i < i^*$. Thus,
    $P = P' \cup \{q\}$ is a set of five points that stabs every disk
    in $\D$, as desired.
\end{proof}
\paragraph{Remark.} A weakness in our proof is that it combines two 
different stages, one
of finding the point $q$ that stabs all the small disks, and one of
constructing the four points of \cref{lem:C:D:diff:size} that stab all
the larger disks. It is an intriguing challenge to merge the two
arguments so that altogether they only require four points. The proof of
Carmi et al.~\cite{CarmiKaMo18} uses a different approach.

\section{Algorithmic Considerations}

The proof of \cref{thm:existence} leads to a simple $O(n \log n)$ time
algorithm for finding a stabbing set of size five.
For this, we need an oracle that decides whether a given 
set of disks is Helly. This has already been done by 
L\"{o}ffler and van Kreveld~\cite{loffler2010largest}, in a 
more general context:

\begin{lemma}[Theorem~6 in \cite{loffler2010largest}]
\label{lem:helly-oracle}
Given a set of $n$ disks, the problem of choosing a point in 
each disk such that the smallest enclosing circle of the resulting 
point set has minimum radius can be solved in $O(n)$ deterministic time.
\end{lemma}

Now, an $O(n\log n)$-time algorithm for finding the five 
stabbing points is based on 
the analysis in the proof of \cref{thm:existence}. It works as follows:
first, we sort the disks in  $\D$ by increasing radius. 
This takes $O(n\log n)$ time. Let 
$\D = \langle D_1, \dots, D_n \rangle$ be the resulting order.
Next, we use binary search with the oracle
from \cref{lem:helly-oracle} to determine the smallest index $i^*$ such that
the prefix $\{D_1, \dots ,D_{i^*}\}$ is non-Helly. 
This yields the disk $D_{i^*}$. We have to invoke 
the oracle $O(\log n)$ times, which
gives a total time of $O(n\log n)$ for this step. After that,
we use another binary search with the oracle 
from \cref{lem:helly-oracle} to determine
the smallest index $k < i^*$ such that 
$\{D_{i^*}, D_1, \dots, D_k\}$ is non-Helly. This costs
$O(n\log n)$ time as well.
Then, we perform a linear search to find an index $j < k$ such that 
$\{D_j,D_k,D_{i^*}\}$
is a non-Helly triple. This step works in $O(n)$ time. Finally, we use 
\cref{lem:helly-oracle} to obtain in $O(n)$ time a stabbing point 
$q$ for the Helly set 
$\{D_1, \dots ,D_{i^* - 1}\}$  and 
the method from the proof of \cref{thm:existence} to 
extend $q$ to a stabbing set for the whole set $\D$. This 
last step works in $O(1)$ time since
the result depends solely on $\{D_j, D_k, D_{i^*}\}$. 
Hence, we can state our claimed theorem.
\begin{theorem}
    Given a set $\D$ of $n$ pairwise intersecting disks in the plane,
    we can find in $O(n\log n)$ time a set $P$ of five points such that
    every disk of $\D$ contains at least one point of $P$.
\end{theorem}

The proof of \cref{lem:helly-oracle} uses the \emph{LP-type framework} by 
Sharir and Welzl~\cite{Chazelle01,SharirWe92}. As we will see next, a
more sophisticated application of the framework directly leads to a 
deterministic linear time algorithm to find a stabbing set with
five points.

\paragraph*{The LP-type framework.}  An \emph{LP-type problem}
$(\H, w, \leq)$ is an abstract generalization of a low-dimensional
linear program.  It consists of a finite set of \emph{constraints}
$\H$, a \emph{weight function} $w: 2^{\H} \rightarrow \W$, and a
\emph{total order} $(\W, \leq)$ on the weights.  The weight function
$w$ assigns a weight to each subset of constraints.  It must fulfill
the following two axioms: 
\begin{itemize}
    \item \textbf{Monotonicity}: for any $\H' \subseteq \H$ and
    $H \in \H$, we have $w\big(\H' \cup \{H\}\big) \leq w(\H')$;
    \item \textbf{Locality:} for any $\B \subseteq \H' \subseteq \H$
    with $w(\B) = w(\H')$ and for any $H \in \H$, we have that if
    $w\big(\B \cup \{H\}\big) = w(\B)$, then also
    $w\big(\H' \cup \{H\}\big) = w(\H')$.
\end{itemize}
Given a subset $\H' \subseteq \H$, a \emph{basis} for
$\H'$ is an inclusion-minimal set $\B \subseteq \H'$ with
$w(\B) = w(\H')$. The \emph{combinatorial dimension} of $(\H,w,\leq)$
is the maximum size of any basis of any subset of $H$.
The goal in an LP-type problem is to determine
$w(\H)$ and a corresponding basis $\B$ for $\H$.
Next, given a set $\B \subseteq \H$ and
a constraint $H \in \H$, we say that $H$ \emph{violates} $\B$ if
$w\big(\B \cup \{H\}\big) < w(\B)$.

A generalization of Seidel's algorithm for low-dimensional linear
programming~\cite{Seidel91,SharirWe92} shows that we can solve an 
LP-type problem
in  $O(|\H|)$ expected time, provided that a constant time algorithm for the
following problem is available. Here and below, the constant factor in the
$O$-notation may depend on the combinatorial dimension.
\begin{itemize}
\item \textbf{Violation test:} Given a basis $\B$ and a constraint $H\in\H$,
determine whether $H$ violates $\B$
and return an error message if $\B$ is not a basis for any
$\H'\subseteq\H$.\footnote{Here, we follow the presentation of
   Chazelle and \Matousek~\cite{ChazelleMa96}. Sharir and 
   Welzl~\cite{SharirWe92}
   use a violation test without the error message. Instead, they need an
   additional \emph{basis computation} primitive: given a basis $\B$
   and a constraint $H \in \H$, find a basis for $\B \cup
   \{H\}$. If a violation test with error message exists and if the 
   combinatorial dimension
   is a constant, a basis computation primitive can easily be
   implemented by brute-force enumeration.}
\end{itemize}

For a deterministic solution, we need an additional computational 
assumption. Let $\B\subseteq\H$ be a basis
of any subset $\H'\subseteq \H$, we use $\vio(\B)$ to denote the set of 
all constraints
$H\in\H$ that violate $\B$, i.e., that have $w(\B\cup\{H\})<w(\B)$. Consider
the \emph{range space} 
$(\H, \Rr=\{\vio(\B) \mid \B \text{ is a basis for some }\H' \subseteq \H\})$.
For a subset $\Y\subseteq\H$, we let $(\Y,\Rr_{\Y})$ be 
the \emph{induced range space}, that is,
$\Rr_{\Y}=\{\Y \cap R \mid R\in\Rr\}$. 
Chazelle and \Matousek~\cite{ChazelleMa96}
have shown that an LP-type problem can be solved in 
$O(\vert\H\vert)$ \emph{deterministic} time
if there is a constant-time violation test as stated above 
and the following computational assumption holds:

\begin{itemize}
\item \textbf{Oracle:} Given a subset $\Y\subseteq\H$, we can 
compute some superset $\Rr'\supseteq\Rr_{\Y}$
in time $\vert\Y\vert^{O(1)}$.
\end{itemize}

During the following discussion, we will show that the problem of 
finding a non-Helly triple as in
\cref{thm:existence} is LP-type and fulfills the four requirements for the
algorithm of Chazelle and \Matousek.

\begin{figure}
\begin{center}
    \includegraphics[scale=0.9]{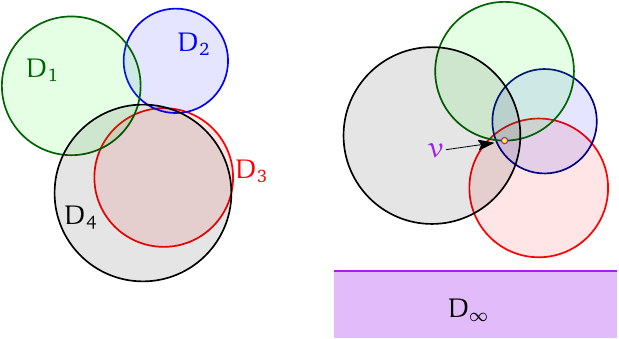}
\end{center}
    \caption[Smallest destroyer and extreme point]{Left: The 
    disks $D_3$ and $D_4$ are destroyers
       of the Helly set $\{D_1,D_2\}$. Moreover, $D_3$ is the smallest
       destroyer of the whole set $\{D_1, D_2, D_3, D_4\}$.
       Right: The disks without $D_ \infty$ form a Helly set
       $\C$. The smallest destroyer of $\C$ is $D_ \infty$ and the
       point $v$ is the extreme point for $\C$ and
       $D_\infty$, i.e., $\dist(\C)=d(v,D_\infty)$.}
    \label{fig:lptype}%
\end{figure}

\paragraph*{Remark.} L\"offler and van Kreveld provide 
proofs that the underlying problem in \cref{lem:helly-oracle} is of
LP-type,
but they do not give arguments for the two computational assumptions, 
see~\cite{loffler2010largest}. However, it is not difficult 
to also verify the two missing statements.

\paragraph*{Geometric observations.}  The \emph{distance} between
two closed sets $A, B \subseteq \R^2$ is defined as
$d(A, B) = \min\ \{\vert ab\vert \mid a \in A, b \in B\}$.  From now on, we
assume that all points in $\bigcup\D$ have positive
$y$-coordinates. This can be ensured with linear overhead by an
appropriate translation of the input.  We denote by $D_{\infty}$ the
closed halfplane below the $x$-axis. It is interpreted as a disk with
radius $\infty$ and center at $(0, -\infty)$. First, observe that for
any subsets $\C_1 \subseteq \C_2 \subseteq \D\cup\{D_\infty\}$, 
we have that if $\C_1$ is non-Helly, then
$\C_2$ is non-Helly. For any  $\C \subseteq \D \cup \{D_\infty\}$,
we say that a disk $D$ \emph{destroys}
$\C$ if $\C \cup \{ D \}$ is non-Helly. Observe that $D_{\infty}$
destroys every non-empty subset of $\D$. Moreover, if $\C$ is
non-Helly, then every disk is a destroyer.
See \cref{fig:lptype} for an example. We can make the
following two observations.

\begin{lemma}
    \label{lem:unique:v}%
    Let $\C \subseteq \D$ be Helly and $D$ a destroyer of $\C$.
    Then, the point $v \in \bigcap\C$ with minimum distance to
    $D$ is unique.
\end{lemma}

\begin{proof}
    Suppose there are two distinct points $v \neq w \in \bigcap \C$
    with $d(v, D) = d\big(\bigcap\C, D \big) = d(w, D)$. Since
    $\bigcap\C$ is convex, the segment $\overline{vw}$ lies in
    $\bigcap\C$. Now, if $D \neq D_\infty$, then every point in the
    relative interior of $\overline{vw}$ is strictly closer to $D$
    than $v$ and $w$. If $D = D_\infty$, then all points in
    $\overline{vw}$ have the same distance to $D$, but since
    $\bigcap \C$ is strictly convex, the relative interior of
    $\overline{vw}$ lies in the interior of $\bigcap \C$, so there
    must be a point in $\bigcap \C$ that is closer to $D$ than $v$ and
    $w$.  In either case, we obtain a contradiction to the assumption
    $v \neq w$ and $d(v, D) = d\big(\bigcap\C, D\big) = d(w, D)$. The
    claim follows.
\end{proof}

Let $\C\subseteq \D$ be Helly and $D$ a destroyer of $\C$. The unique point
$v \in \bigcap \C$ with minimum distance to $D$ is called the
\emph{extreme point} for $\C$ and $D$ (see \cref{fig:lptype}, right).

\begin{figure}
\begin{center}
\includegraphics[scale=0.9]{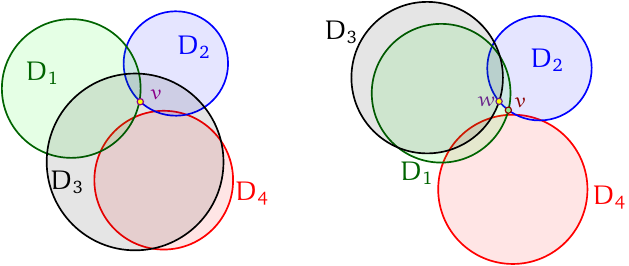}%
\end{center}
\caption[Illustration of \cref{lem:destroyer:dist}]{Left: The disk $D_4$ is 
a destroyer for the
Helly sets $\{D_1, D_2\}$ and $\{D_1, D_2, D_3\}$.  The extreme
point $v$ for $\{D_1, D_2\}$ is also the extreme point for
$\{D_1, D_2, D_3\}$.  Right: The disk $D_4$ is a
destroyer for the Helly sets $\{D_1, D_2\}$ and
$\{D_1, D_2, D_3\}$.  The extreme point $v$ for $\{D_1, D_2\}$
is not in $D_3$. The distance to $D_4$ increases.}
\label{fig:destroyers}%
\end{figure}

\begin{lemma}
    \label{lem:destroyer:dist}%
    Let $\C_1 \subseteq \C_2 \subseteq\D$ be two Helly sets and
    $D$ a destroyer of $\C_1$ (and thus of $\C_2$).  Let
    $v \in \bigcap\C_1$ be the extreme point for $\C_1$ and $D$. We
    have $d\big(\bigcap\C_1, D\big) \leq d\big(\bigcap\C_2, D\big)$.
    In particular, if $v \in \bigcap \C_2$, then
    $d\big(\bigcap\C_1, D\big) = d\big(\bigcap\C_2, D\big)$ and $v$ is
    also the extreme point for $\C_2$ and $D$.  If $v \not\in \bigcap \C_2$,
    then $d\big(\bigcap\C_1, D\big) < d\big(\bigcap\C_2, D\big)$.
\end{lemma}
\begin{proof}
    The first claim holds trivially: let $w \in \bigcap\C_2$ be the
    extreme point for $\C_2$ and $D$.  Since $\C_1 \subseteq \C_2$, it
    follows that $w \in \bigcap\C_1$, so
    $d\big(\bigcap \C_1, D \big) \leq d(w, D) = d\big(\bigcap\C_2,
    D\big)$.  If $v \in \bigcap\C_2$, then
    $d\big(\bigcap\C_1, D\big) \leq d\big(\bigcap\C_2, D\big) \leq
    d(v, D) = d\big(\bigcap\C_1, D\big)$, so $v = w$, by
    \cref{lem:unique:v}. If $v \notin\bigcap\C_2$, then
    $d\big(\bigcap\C_1, D\big) < d\big(\bigcap\C_2, D\big)$, by
    \cref{lem:unique:v} and the fact that $\C_1 \subseteq \C_2$.
    See \cref{fig:destroyers}.
\end{proof}

Let $\C$ be a subset of $\D$. For $0 < r \leq \infty$ we
define $\C_{<r}$ as the set of all disks in
$\C$ with radius smaller than $r$. Recall that we assume
that all the radii are pairwise distinct.
A disk $D$ with radius $r$, $0 < r \leq \infty$,  is called
\emph{smallest destroyer} of $\C$ if (i) $D \in \C$ or $D=D_{\infty}$,
(ii) $D$ destroys $\C_{<r}$, and (iii) there is no disk $D'\in\C_{<r}$
that destroys $\C_{<r}$. Observe that Property~(iii)
is the same as saying that $\C_{<r}$ is Helly. 
See \cref{fig:lptype} for an example.

Let $\C$ be a subset of $\D$ and $D$ the 
smallest destroyer of $\C$.
We write $\rad(\C)$ for the radius of $D$
and $\dist(\C)$ for
the distance between $D$ and the set $\bigcap\C_{<\rad(\C)}$, i.e.,
$\dist(\C)= d\big(\bigcap\C_{<\rad(\C)}, D\big)$.  Now, if $\C$ is Helly,
then $D=D_{\infty}$ and thus $\rad(\C) = \infty$.
If $\C$ is non-Helly,
then $D \in \C$ and thus $\rad(\C)<\infty$.
In both cases, $\dist(\C)$ is the distance between $D$ and
the extreme point for $\C_{<\rad(\C)}$ and $D$. We define the
\emph{weight of $\C$} as $w(\C) = (\rad(\C),-\dist(\C))$, and
we denote by $\leq$ the lexicographic order on $\R^2$.  Chan
observed, in a slightly different context, that $(\D, w ,\leq)$ is
LP-type~\cite{Chan04}.  However, Chan's paper does not contain a
detailed proof for this fact. Thus, in the following lemmas, we show
the two LP-type axioms, present a constant time violation test,
and a polynomial-time oracle. We start with the monotonicity axiom 
followed by the locality axiom.
\begin{figure}
    \centering %
    \includegraphics[scale=0.9]{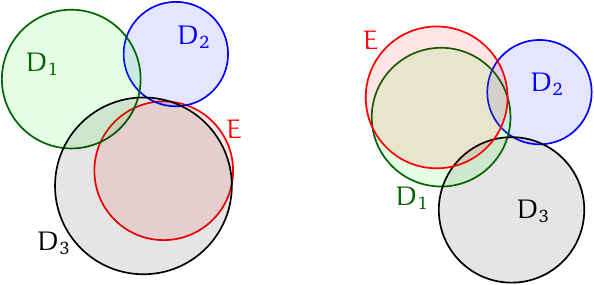}%
    \caption[The LP-type axiom monotonicity]{Monotonicity: In both cases, $\{D_1, D_2, D_3\}$ is
       non-Helly with smallest destroyer $D_3$.  Adding a disk $E$
       either decreases the radius of the smallest destroyer
       (left) or increases the distance to the smallest
       destroyer (right).}
    \label{fig:axiom1}%
\end{figure}
\begin{lemma}%
    \label{lem:LPtype:axiom1}%
    For any $\C \subseteq \D$ and $E \in \D$, we have
    $w\big(\C \cup\{E\}\big) \leq w(\C)$.
\end{lemma}
\begin{proof}
    Set $\C^* = \C \cup \{E\}$. Let $D$ be the smallest destroyer of
    $\C$, and let $r = \rad(\C)$ be the radius of $D$.
    Since $D$ destroys $\C_{< r}$, the set $\C_{< r}\cup\{D\}$ is
    non-Helly. Moreover, since $\C_{< r}\cup\{D\}\subseteq \C^*_{< r}\cup\{D\}$,
    we know that $\C^*_{< r}\cup\{D\}$ is also non-Helly.
    Therefore, $D$ destroys $\C^*_{< r}$ and we can derive
    $\rad(\C^*) \leq \rad(\C)$. If we have
    $\rad(\C^*) < \rad(\C)$, we are done. Hence, assume that
    $\rad(\C^*) = \rad(\C)$. Then $D$ is the smallest destroyer of $\C^*$,
    and \cref{lem:destroyer:dist} gives
    $-\dist(\C^*) = -d\big(\bigcap\C^*_{<r}, D\big) \leq -
    d(\bigcap\C_{<r}, D) = -\dist(\C)$. Hence,
    $w\big(\C^*) \leq w(\C)$. See \cref{fig:axiom1} for an
    illustration.
\end{proof}

\begin{lemma}
    \label{lem:LPtype:axiom3}%
    Let $\B \subseteq \C \subseteq\D$ with $w(\B) = w(\C)$ and let
    $E \in \D$. Then, if $w \big(\B \cup \{E\}\big) = w(\B)$, we also
    have $w\big(\C \cup \{E\}\big) = w(\C)$.
\end{lemma}

\begin{proof}
    Set $\C^* = \C \cup \{E\}$, $\B^* = \B\cup \{E\}$.  Let
    $r = \rad(\C)$ and $D$ be the smallest destroyer of $\C$.
    Since $w(\C) = w(\B) = w(\B^*)$, we have that $D$ is also the
    smallest destroyer of $\B$ and of $\B^*$.
    If the radius of $E$ is larger than $r$, then $E$ cannot be the smallest
    destroyer of $\C^*$, so $w\big(\C^*\big) = w(\C)$.
    Thus, assume that  $E$ has radius less than $r$.  Let $v$ be the 
    extreme point of
    $\C_{<r}$ and $D$. Since $w(\B^*) = w(\B)$, we know that
    $d\big(\bigcap\B_{<r}, D\big) = d\big(\bigcap\B^*_{<r}, D\big) =
    d(v, D)$.  Now, \cref{lem:destroyer:dist} implies that $v \in E$,
    since $E \in \B^*_{<r}$.  Thus, the set
    $\C^*_{<r} = \C_{<r} \cup \{E\}$ is Helly and therefore, there is no
    disk $D'\in \C^*_{<r}$ that destroys $\C^*_{<r}$.
    Furthermore, since $D$ destroys $\C_{<r}$ and
    $\C_{<r} \subset \C^*_{<r}$, the disk $D$ also destroys $\C^*_{<r}$.
    Therefore, $D$ is also the smallest destroyer of $\C^*$, so
    $\rad(\C^*) = r = \rad(\C)$.  Finally, since
    $\B^*_{<r} \subseteq \C^*_{<r}$ we can use
    \cref{lem:destroyer:dist} to derive
    \[
        d\Big(\bigcap\C_{<r}, D\Big) = d\Big(\bigcap\B^*_{<r}, D\Big)
        \leq d\Big(\bigcap\C^*_{<r}, D\Big) \leq d(v, D) =
        d\Big(\bigcap\C_{<r}, D\Big).
    \]
    The claim follows.
\end{proof}

Next, we are going to describe the violation test for $(\D, w, \leq)$: given 
a basis
$\B \subseteq \D$ and a disk $E \in \D$, check whether $E$
violates $\B$, i.e., whether $w\big(\B \cup \{E\}\big) < w(\B)$, and 
return an error message
if $\B$ is not a basis. But first, we show that the combinatorial dimension of
$(\D, w, \leq)$ is at most $3$.

\begin{figure}
    \centering%
    \includegraphics[scale=0.9]{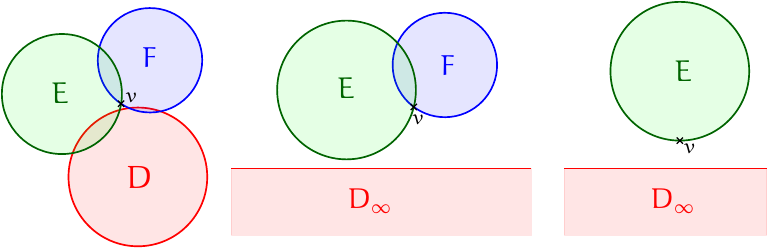}%
    \caption[The range space has a finite basis]{A basis can either be a non-Helly triple (left), a pair
       of intersecting disks $E$ and $F$ where the point of minimum
       $y$-coordinate in $E \cap F$ is a vertex (middle), or a single
       disk (right).}
    \label{fig:axiom2}%
\end{figure}
\begin{lemma}%
    \label{lem:LPtype:axiom2}%
    For each $\C \subseteq \D$, there is a set $\B \subseteq \C$ with
    $|\B| \leq 3$ and $w(\B) = w(\C)$.
\end{lemma}

\begin{proof}
    Let $D$ be the smallest destroyer of $\C$.  Let
    $r = \rad(\C)$ be the radius of $D$, and let
    $v \in \bigcap \C_{<r}$ be the extreme point for $\C_{<r}$ and
    $D$. 
    First of all, we observe that $v$ cannot be in the
    interior of $\bigcap \C_{<r}$, since $v$ minimizes the distance to $D$.
    Thus, there has to be a non-empty
    subset $\A \subseteq \C_{<r}$ such that $v$ lies on the boundary of
    each disk of $\A$. 
    Let  $\A$ be a minimal set such that $d(\bigcap \A, D) = d(v, D)$.
    It follows that $|\A| \leq 2$.
    See \cref{fig:axiom2} for an illustration.
    
    First, assume that $\A = \{E\}$. Then, since
    $d(E, D) = d(v,D)>0$, we know that $E\cap D=\emptyset$.
    As the disks in $\C$ intersect pairwise, 
    we derive $D\notin \C$ and hence $D=D_\infty$. Setting
    $\B= \A$, we get $\rad(\C)=\infty=\rad(\B)$ and
    $\dist(\C)=d(v,D)=d(E,D)=\dist(\B)$.
    Thus, $|\B|\leq 3$ and $w(\B)=w(\C)$.
    
    Second, assume that $\A=\{E,F\}$.
    Then, $v$ is one of the two vertices of the lens $L= E \cap F$.
    Next, we show that $d(L, D)\geq d(v, D)$. Assume for the sake of
    contradiction that there is a point $w \in L$ with $d(w, D) < d(v, D)$.
    By general position and since $v$ is the intersection
    of two disk boundaries, there is a relatively open neighborhood
    $N$ around $v$ in $\bigcap\C_{<r}$ such that $N$ is also
    relatively open in $L$.  Since $L$ is convex, there
    is a point $x \in N$ that also lies in the relative interior of
    the line segment $\overline{wv}$. Then, $d(x, D) < d(v, D)$ and
    $x \in \bigcap\C_{<r}$. This yields a contradiction, as $v$ is the extreme
    point for $\C_{<r}$ and $D$. Thus, we have $d(L, D)\geq d(v, D)$
    which also shows hat $D \cap E \cap F= \emptyset$.
    
    We set $\B = \{E,F\}$, if $\C$ is Helly (i.e., $D=D_\infty$), and
    $\B = \{D, E, F\}$, if $\C$ is non-Helly (i.e., $D \in \C$). In both cases,
    we have $\B \subseteq \C$ and $|\B| \leq 3$. Moreover, we can conclude
    that $D$ destroys $\B_{<r} = \{E,F\}$, and since $\B_{<r}$ is Helly, $D$
    is the smallest destroyer of $\B$. Hence, we have $\rad(\C) = r = \rad(\B)$.

    To obtain $\dist(\B) = \dist(\C)$, it remains to show
    $d(\bigcap\B_{<r},D)=d(\bigcap\C_{<r},D)$.  Since $\B_{<r} \subseteq \C_{<r}$,
    we can use \cref{lem:destroyer:dist} as well as $d(L, D)\geq d(v, D)$ to
    derive
    \[
		d\Big(\bigcap\C_{<r},D\Big) \geq d\Big(\bigcap\B_{<r},D\Big),
		= d(L,D) \geq d(v,D) = d\Big(\bigcap\C_{<r},D\Big)
    \]
    as desired. We conclude that $w(\B)=w(\C)$.    
    
    We remark that the set $\B$ is actually a basis for $\C$: if $\B$ is a non-Helly
    triple, then removing any disk from $\B$ creates a Helly set and
    increases the radius of the smallest destroyer to $\infty$.  If
    $|\B| \leq 2$, then $D_\infty$ is the smallest destroyer of $\B$
    and the minimality follows directly from the definition.
\end{proof}

Following the argument of the last proof, the violation test is now immediate.
We present pseudo-code in Algorithm~\ref{alg:violationTest}.
It obviously needs constant time. Finally, to apply
the algorithm of Chazelle and \Matousek, we still need to check
that there is a polynomial-time oracle that computes a superset of 
$\Rr_{\Y}$ for a given
set of disks $\Y$.
\begin{algorithm}
\caption{The violation test.}
\begin{algorithmic}[1]
\Procedure {violates}{set $\B\subseteq\D$, disk $E\in\D$ with radius $r'$}
\If {$|\B| > 3$ or $|\B| = 3$ and $\B$ is Helly} 
\Return \textbf{``{\boldmath $\B$} is not a basis.''}
\EndIf
\If {$|\B| = 2$ and the $y$-minimum of $\bigcap \B$ is also
    the $y$-minimum of a single disk of $\B$}
\State \Return \textbf{``{\boldmath $\B$} is not a basis.''}
\EndIf
\If {$\B=\{D_1\}$}
\If {the $y$-minimum in $E \cap D_1$ differs from the $y$-minimum in $D_1$}
\State \Return \textbf{``{\boldmath $E$} violates {\boldmath $\B$}.''}
\Else\ \Return \textbf{``{\boldmath $E$} does not violate {\boldmath $\B$}.''}
\EndIf
\EndIf
\If {$\B=\{D_1,D_2\}$}
\State $v=\argmin\ \{w_y\mid w\in D_1\cap D_2\}$
\If {$v\notin E$}
\Return \textbf{``{\boldmath $E$} violates {\boldmath $\B$}.''}
\Else\ \Return \textbf{``{\boldmath $E$} does not violate {\boldmath $\B$}.''}
\EndIf
\Else
\algorithmiccomment{$\B$ is of size $3$, non-Helly, and does not 
contain $D_\infty$.}
\State $D=$ smallest destroyer of $\B$
\State $\{D_1,D_2\}=\B\setminus\{D\}$
\State $r=\rad(\B)$
\If{$r'>r$}
\Return \textbf{``{\boldmath $E$} does not violate {\boldmath $\B$}.''}
\Else
\State $v=\argmin\ \{d(w,E)\mid w\in D_1\cap D_2\}$
\If {$v\notin E$}
\Return \textbf{``{\boldmath $E$} violates {\boldmath $\B$}.''}
\Else\ \Return \textbf{``{\boldmath $E$} does not violate {\boldmath $\B$}.''}
\EndIf
\EndIf
\EndIf
\EndProcedure
\end{algorithmic}
\label{alg:violationTest}
\end{algorithm}

\begin{lemma}
    \label{lem:oracle}%
    Given a set $\Y\subseteq\D$ of disks, we can compute a 
    superset of $\Rr_{\Y}$
    in time $O(\vert\Y\vert^4)$.
\end{lemma}

\begin{proof}
Let $v\in\R^2$ and $r>0$.
First, we let $R_v=\{D\in\Y\mid v\notin D\}$ be the range of all disks
that do not contain 
$v$. Second, let $R_{v,r}$ be the range of all disks of diameter 
smaller than $r$
that do not contain  the point $v$, i.e.,
$R_{v,r}=\{D\in\Y\mid v\notin D\text{ and } r_D<r\}$. 
We define $\Rr'$ to be the set of
all ranges $R_v$ over all $v$ and subsequently, we let $\Rr''$ be the set 
of all ranges $R_{v,r}$ over all $v$ and $r$, that is, 
$\Rr''=\{R_{v,r}\mid v\in\R^2\text{ and } r>0\}$.

The discussion from the previous lemmas shows that for any basis $\B$, 
there is a
    point $v_\B \in \R^2$ and a radius $r_\B>0$ such that a disk $E \in \D$
    with radius $r_E$ violates $\B$ if
    and only if $v_\B \not\in E$ and $r_E<r_\B$. Hence, we have 
    $\Rr''\supseteq\Rr_{\Y}$.
    We show how to compute $\Rr''$ in polynomial time. For this, we 
    first construct $\Rr'$.

For the given set $\Y$ of disks, we compute the arrangement $A(\Y)$ and then 
focus on the facets
of $A(\Y)$. Since the arrangement has $O(\vert\Y\vert^2)$ 
facets, we can compute $A(\Y)$ in time $O(\vert\Y\vert^3)$ 
using a simple brute-force approach (faster algorithms exist, 
but are not needed here).
Clearly, for two points $v$ and $w$ of the same facet of $A(\Y)$, 
we have $R_v=R_w$. Therefore, for a given facet $f$, 
we pick an arbitrary point $v\in f$, and we compute $R_v$ by a linear scan 
of $\Y$.
Summing over all facets, we can thus compute $\Rr'$ in time $O(\vert\Y\vert^3)$.
    
Finally, to compute $\Rr''$, we iterate over all $O(\vert\Y\vert^2)$ 
ranges in $\Rr'$. Given a range $R_v\in\Rr'$,
we get all $R_{v,r}$ for $r>0$ by first sorting $R_v$ by increasing radii 
and then taking every prefix of the
sorted list of disks. For a fixed $v$, this can be done in time 
$O(\vert\Y\vert^2)$. Hence, $\Rr''$ can be computed in 
$O(\vert\Y\vert^4)$ time. The claim follows.
\end{proof}
The following lemma summarizes the discussion so far.
\begin{lemma}%
    \label{lem:solve:LPtype}%
    Given a set $\D$ of $n$ pairwise intersecting disks in the plane,
    we can decide in $O(n)$ deterministic time whether $\D$ is Helly.
    If so, we can compute a point in $\bigcap\D$ in $O(n)$
    deterministic time.  If not, we can compute the smallest destroyer
    $D$ of $\D$ and two disks $E, F \in \D_{<r}$ that form a non-Helly
    triple with $D$.  Here, $r$ is the radius of $D$.
\end{lemma}

\begin{proof}
    Since (i) $(\D, w, \leq)$ is LP-type, (ii) the violation test needs 
    constant
    time, and (iii) the oracle needs polynomial time, we can apply
    the deterministic algorithm of Chazelle and
    \Matousek~\cite{ChazelleMa96} to compute
    $w(\D) = (\rad(\D), -\dist(\D))$ and a corresponding basis $\B$ in
    $O(n)$ time.  Then, $\D$ is Helly if and only if
    $\rad(\D) = \infty$. If $\D$ is Helly, then $|\B| \leq 2$. We
    compute the unique point $v \in \bigcap\B$ with
    $d(v, D_\infty) = d\big(\bigcap\B, D_{\infty}\big)$.  Since
    $\B \subseteq \D$ and
    $d\big(\bigcap\B,D_{\infty}\big) = d\big(\bigcap\D,
    D_{\infty}\big)$, we have $v \in \bigcap\D$ by
    \cref{lem:destroyer:dist}. We output $v$.  If $\D$ is non-Helly, we
    simply output $\B$, because $\B$ is a non-Helly triple with the
    smallest destroyer $D$ of $\D$ and two disks $E, F \in\D_{<r}$,
    where $r$ is the radius of $D$.
\end{proof}

\begin{theorem}
    Given a set $\D$ of $n$ pairwise intersecting disks in the plane,
    we can find in deterministic $O(n)$ time a set $P$ of five points such that
    every disk of $\D$ contains at least one point of $P$.
\end{theorem}

\begin{proof}
    Using the algorithm from \cref{lem:solve:LPtype}, we decide whether
    $\D$ is Helly.  If so, we return the extreme point computed by the
    algorithm.  Otherwise, the algorithm gives us a non-Helly triple
    $\{D, E, F\}$, where $D$ is the smallest destroyer of $\D$ and
    $E, F \in\D_{<r}$, with $r$ being the radius of $D$.  Since
    $\D_{<r}$ is Helly, we can obtain in $O(n)$ time a stabbing point
    $q \in \bigcap\D_{<r}$ by using the algorithm from
    \cref{lem:solve:LPtype} again.  Next, by \cref{lem:triple}, there are
    two disks in $\{D, E, F\}$ whose lens angle is at least
    $2\pi/3$. Let $P'$ be the set of four points from the proof of
    \cref{lem:C:D:diff:size}.  Then, $P = P' \cup \{q\}$ is a set of
    five points that stabs every disk in $\D$.
\end{proof}

\section{Simple Bounds}
\label{sec:simple_bounds}

We now provide some easy lower and upper bounds on the number of disks 
for which a
certain number of stabbing points is necessary or sufficient.
\paragraph{Eight disks can be stabbed by three points.} 
For the proof that any set of eight pair-wise intersecting disks can
be stabbed by at most three points, we show the following lemma.
\begin{lemma}
\label{lem:5-disks}
Let $\D$ be a set of at least $5$ pairwise intersecting disks.
Then, $\D$ contains a Helly-triple.
\end{lemma}
\begin{proof}
Let $\D$ be a set of exactly $5$ pairwise intersecting disks.
We assume that no three centers of the disks are on a line, since otherwise
these three disks are a Helly-triple.
Since the complete graph $K_5$ does not have a planar embedding, there have
to be four different disks $D_1,\dots,D_4\in\D$ with centers 
$c_1,\dots, c_4$ and
radii $r_1,\dots,r_4$ such that
the line segments $c_1c_3$ and $c_2c_4$ intersect, see \cref{fig:4_disks}. 
Let $x$ be the intersection point.
\begin{figure}
    \centering%
    \includegraphics[scale=1, trim = 45 50 50 30,clip]{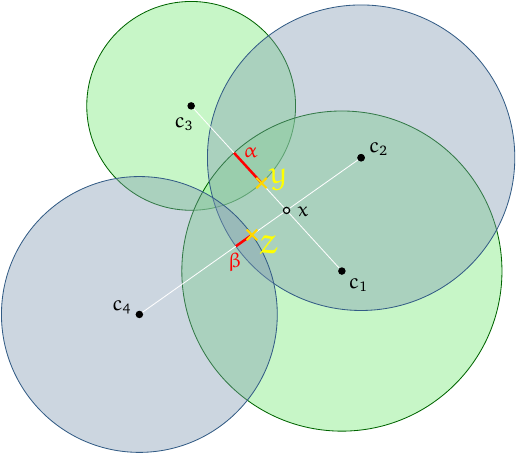}%
        \caption[Illustration of \cref{lem:5-disks}]{Proof of \cref{lem:5-disks}.}%
        \label{fig:4_disks}
\end{figure}
Moreover, let $\alpha$ (resp., $\beta$) be the intersection of 
the lens $L_{1,3}$
(resp., $L_{2,4}$) and the line segment $c_1c_3$ (resp., $c_2c_4$). If $x$ is in
$\alpha$ or $\beta$, we are done. Otherwise, let $y$ be the point of $\alpha$
that is closest to $x$ and let $z$ be the point of $\beta$ closest to $x$.
We can assume without loss of generality that $\vert xy\vert\leq\vert xz\vert$
and $x\notin D_4$.
Using the triangle inequality, We can derive 
\[
\vert c_2y\vert
\leq \vert c_2x\vert +\vert xy\vert
\leq \vert c_2x\vert+ \vert xz\vert \leq r_2
\]
to conclude that $y\in D_1\cap D_2\cap D_3$.
\end{proof}

Now consider a set $\D$ of $8$ pairwise intersecting disks. Using 
\cref{lem:5-disks},
we can find a Helly-triple in $\D$. Among the remaining $5$ disks, 
we find a second
Helly-triple. The remaining two disks can be stabbed by one point. 
This reasoning
yields the following corollary, which was already mentioned by 
\Stacho~\cite{Stacho65}.

\begin{corollary}
Every set $\D$ of at most $8$ pairwise intersecting disks 
can be stabbed by 3 points.
\end{corollary}

\paragraph{13 disks with 4 stabbing points.}
Danzer presented a set of $10$ pairwise intersecting pseudo-disks with 
stabbing number
four \cite{Danzer86}. However,
it is not clear to us how these $10$ pseudo-disks can be realized as pairwise 
intersecting Euclidean disks
achieving the same stabbing number. Moreover, it is another 
open problem whether $9$
pairwise intersecting disks can be stabbed by three points. 
Instead, we want to describe a
set of $13$ pairwise intersecting disks in the plane
such that no point set of size three can pierce all of them.

 The
construction begins with an inner disk $A$ of radius $1$ and three
larger disks $D_1$, $D_2$, $D_3$ of equal radius, so that each pair  
of disks in $\{A,D_1,D_2,D_3\}$ is tangent.
For $i = 1, 2, 3$, we denote the contact point of $A$ and
$D_i$ by $\xi_i$.

We add six more disks as follows. For $i=1,2,3$, we draw
    the two common outer tangents to $A$ and $D_i$, and denote by
    $T_i^-$ and $T_i^+$ the halfplanes that are bounded by these
    tangents and are openly disjoint from $A$. The labels $T_i^-$ and
    $T_i^+$ are chosen such that the points of tangency between $A$
    and $T_i^-$, $D_i$, and $T_i^+$, appear along the boundary of $A$ in this
    counterclockwise order.  One can show that the nine points of
    tangency between $A$ and the other disks and tangents are pairwise
    distinct (see \cref{fig:lower}).
    We regard the six halfplanes
    $T_i^-$, $T_i^+$, for $i=1,2,3$, as (very large) disks; in the
    end, we can apply a suitable inversion to turn the disks and
    halfplanes into actual disks, if so desired.

\begin{figure}
    \centering%
    \includegraphics[scale=0.7]{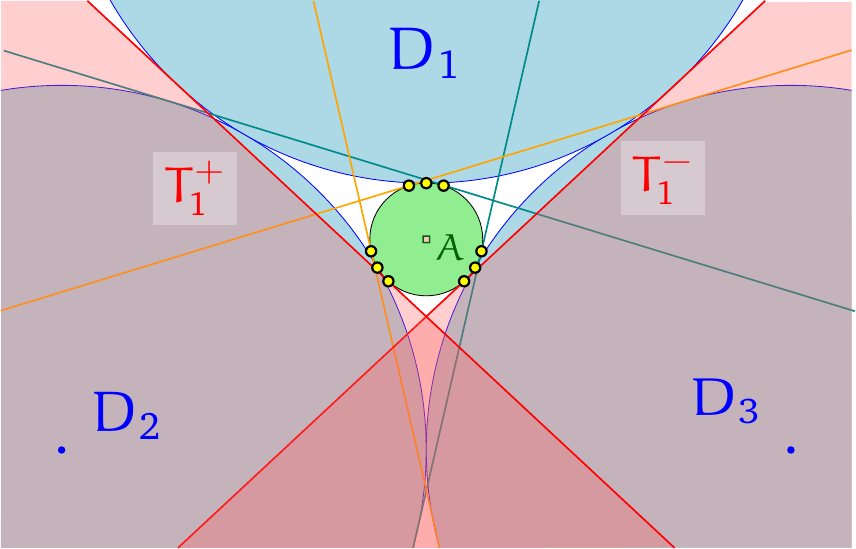}%
        \caption[Lower bound construction]{Each common tangent $\ell$ between $A$ and
           $D_i$ represents a very large disk, whose interior is 
           disjoint from $A$. The nine points of
           tangency are pairwise distinct.}%
        \label{fig:lower}
\end{figure}

Finally, we construct three additional disks $A_1$, $A_2$, $A_3$.  To
construct $A_i$, we slightly expand $A$ into a disk $A'_i$ of radius
$1 + \eps_1$, while keeping the tangency with $D_i$ at $\xi_i$.  We
then roll $A'_i$ clockwise along $D_i$, by a tiny angle
$\eps_2 \ll \eps_1$, to obtain $A_i$.

This gives a set of $13$ disks.  For sufficiently small $\eps_1$ and
$\eps_2$, we can ensure the following properties for each $A_i$: (i)
$A_i$ intersects all other $12$ disks; (ii) the nine intersection
regions $A_i \cap D_j$, $A_i \cap T_j^-$, $A_i \cap T_j^+$, for
$j = 1,2,3$, are pairwise disjoint; and (iii) $\xi_i\notin A_i$.

\begin{theorem}%
    \label{thm:lower:bound}
    The construction yields a set of $13$ disks that cannot be stabbed
    by $3$ points.
\end{theorem}

\begin{proof}
    Consider any set $P$ of three points.  Set
    $A^* = A \cup A_1 \cup A_2 \cup A_3$.  If
    $P \cap A^* = \emptyset$, we have unstabbed disks, so suppose that
    $P \cap A^* \neq \emptyset$. For $p \in P \cap A^*$, property~(ii)
    implies that $p$ stabs at most one of the nine remaining disks
    $D_j$, $T_j^+$ and $T_j^-$, for $j = 1, 2, 3$.  Thus, if
    $P \subset A^*$, we would have unstabbed disks, so we may assume
    that $|P \cap A^*| \in \{1, 2\}$.

    Suppose first that $|P \cap A^*| = 2$. As just argued, at most two
    of the remaining disks are stabbed by $P \cap A^*$. The following
    cases can then arise.
    \begin{enumerate}[(a)]
        \item None of $D_1$, $D_2$, $D_3$ is stabbed by $P \cap A^*$.
        Since $\{D_1, D_2, D_3\}$ is non-Helly and a non-Helly set
        must be stabbed by at least two points, at least one disk
        remains unstabbed.
        \item Two disks among $D_1$, $D_2$, $D_3$ are stabbed by
        $P \cap A^*$.  Then the six unstabbed halfplanes form many
        non-Helly triples, e.g., $T_1^-$, $T_2^-$, and $T_3^-$, and
        again, a disk remains unstabbed.
        \item The set $P \cap A^*$ stabs one disk in
        $\{D_1, D_2, D_3\}$ and one halfplane. Then, there is (at
        least) one disk $D_i$ such that $D_i$ and its two tangent
        halfplanes $T_i^-$, $T_i^+$ are all unstabbed by $P \cap A^*$.
        Then, $\{ D_i, T_i^-, T_i^+ \}$ is non-Helly, and at least $2$
        more points are needed to stab it.
    \end{enumerate}
    Suppose now that $|P \cap A^*| = 1$, and let
    $P \cap A^* = \{p\}$.  We may assume that $p$ stabs all four
    disks $A$, $A_1$, $A_2$, $A_3$, since otherwise a disk would stay
    unstabbed. By property (iii), we can derive $p \not\in \{\xi_1, \xi_2, \xi_3\}$.
    Now, since $p\in A \setminus \{\xi_1, \xi_2, \xi_3\}$, the point $p$ does
    not stab any of $D_1$, $D_2$, $D_3$. Moreover, by property (ii),
    the point $p$ can only stab at most one of the remaining halfplanes.
    Since $\{D_1, D_2, D_3\}$ is
    non-Helly, it requires two stabbing points. Moreover, since
    $|P \setminus \{p\}| = 2$, it must be the case that one point $q$ of
    $P \setminus A^*$ is the point of tangency of two of these disks,
    say $q = D_2 \cap D_3$.  Then, $q$ stabs only two of the six
    halfplanes, say, $T_1^-$ and $T_1^+$. But then,
    $\{ D_1, T_2^+, T_3^- \}$ is non-Helly and does not contain any
    point from $\{p, q\}$. At least one disk remains unstabbed.
\end{proof}

\section{Conclusion}

We gave a simple linear-time algorithm, based on techniques for solving 
LP-type problems,
to find five stabbing points
for a set of pairwise intersecting disks in the plane.
The arXiv manuscript by Carmi, Katz, and Morin~\cite{CarmiKaMo18}
claims a similar linear-time algorithm for finding four stabbing points.
It would now be interesting to see whether these results,
the ones by Danzer, \Stacho, and ours, could be used to find
new deterministic approximation algorithms for computing large cliques in 
disk graphs;
refer to~\cite{ambuhl2005clique, bonamy2018eptas} for the known algorithms.
On the lower-bound side, it is still not known
whether nine disks can always be stabbed by three points or not. For eight
disks, we provided a proof that three points always suffice, 
as already mentioned by \Stacho~\cite{Stacho65}.
The lower bound construction of Danzer with ten 
disks~\cite{Danzer86}
can easily be verified for pseudo-disks. However, 
the example is not easy to draw, even with 
the help of geometry processing
software. Until now, we were not able to check whether his 
pseudo-disk arrangement can
be realized as a Euclidean disk arrangement.

\bibliographystyle{abbrv}
\bibliography{5_pts}

\end{document}